\documentclass[prl,notitlepage, reprint]{revtex4-1}
\usepackage[utf8]{inputenc}
\usepackage{amsmath}
\usepackage{braket}
\usepackage{bbold}
\usepackage{graphicx}
\usepackage{amsthm}
\usepackage{mathrsfs}
\usepackage{hyperref}

\newtheorem{theorem}{Theorem}
\newtheorem{conjecture}{Conjecture}

\usepackage{color}

\hypersetup{
    colorlinks=true,
    citecolor=blue
}

\newcommand{\pp}[1]{\ensuremath \mathbb{#1}}
\newcommand{\mup}{\ensuremath \mu_{\mathbb{P}}}
\newcommand{\HH}{\ensuremath \mathcal{H}}
\newcommand{\XX}{\ensuremath \mathcal{X}}
\newcommand{\dd}{\ensuremath \,\mathrm{d}}

\begin{document}
\title{Quantum Mean Embedding of Probability Distributions}
\date{\today}
\author{Jonas M. Kübler}
\email{jmkuebler@tuebingen.mpg.de}
\affiliation{Max Planck Institute for Intelligent Systems, T\"ubingen 72076, Germany}
\author{Krikamol Muandet}
\email{krikamol@tuebingen.mpg.de}
\affiliation{Max Planck Institute for Intelligent Systems, T\"ubingen 72076, Germany}
\author{Bernhard Sch\"olkopf}
\email{bs@tuebingen.mpg.de}
\affiliation{Max Planck Institute for Intelligent Systems, T\"ubingen 72076, Germany}

\begin{abstract}
    The kernel mean embedding of probability distributions is commonly used in machine learning as an injective mapping from distributions to functions in an infinite dimensional Hilbert space. It allows us, for example, to define a distance measure between probability distributions, called maximum mean discrepancy (MMD). 
    In this work we propose to represent probability distributions in a pure quantum state of a system that is described by an infinite dimensional Hilbert space. 
    This enables us to work with an explicit representation of the mean embedding, whereas classically one can only work implicitly with an infinite dimensional Hilbert space through the use of the kernel trick. We show how this explicit representation can speed up methods that rely on inner products of mean embeddings and discuss the theoretical and experimental challenges that need to be solved in order to achieve these speedups.
\end{abstract}

\maketitle
\section{Introduction}
In machine learning, kernel methods are used to implicitly evaluate inner products in high dimensional feature spaces. 
Popular algorithms such as the support vector machine \cite{Cortes95:SN,Steinwart08:SVM} or principal component analysis \cite{Hotelling33:PCA}, which are linear methods, can be expressed solely in terms of inner products between data points.
These methods become more expressive if the data is first mapped onto a high dimensional feature space.
Instead of evaluating the inner product explicitly in the feature space, whose cost scales linearly with the feature space dimension, a more efficient evaluation can be done implicitly in the original space using a positive definite kernel function. 
This is known as the \textit{kernel trick} \cite{Schoelkopf2001}. 
Since it does not require an explicit feature map, the kernel trick even allows us to work with infinite dimensional feature spaces, e.g., using a Gaussian kernel. 
The downside of most kernel-based methods is that they scale polynomially with the size of the data sets. This problem has been tackled in the realm of quantum computation and exponential speedups have been conjectured \cite{Rebentrost2014, Lloyd2014}.
Such speedups are, however, still  highly controversial \cite{Aaronson2015, Ciliberto2018}. \\
Only recently has the cost of a single kernel evaluation been the target of quantum computing research \cite{Chatterjee2017, Schuld2019, havlivcek2019}. 
Speedups might be possible, since the cost of explicitly evaluating inner products of quantum states only grows logarithmically with the system size \cite{Cincio2018}, as opposed to linear on a classical computer. 
Schuld and Killoran further conjecture the usage of continuous variable quantum systems for working with classically intractable, i.e., hard to compute, kernels in infinite dimensions \cite{Schuld2019}, but it is unclear whether problems exist for which such kernels can lead to an improvement. 
Furthermore, the recent suggestions do not address the polynomial scaling of kernel methods with the sample size, leaving the application of quantum computing in large-scale kernel methods a challenging problem.\\
The idea of explicitly representing an infinite dimensional feature vector as a quantum state opens a way to tackle this problem. 
While it is impossible classically to sum two infinite dimensional vectors, a quantum mechanical \textit{superposition} of two states can be constructed explicitly, even for infinite dimensional systems, see, e.g., \cite{Vlastakis2013}. 
On the other hand, \emph{the evaluation of inner products in an infinite dimensional quantum Hilbert space is independent of the number of states in a superposition}. 
We identify methods involving the \textit{kernel mean embedding} \cite{Berlinet2004, Smola2007, Muandet2017} as a branch of machine learning techniques that suffer from the fact that  on a classical computer the cost of the evaluation of inner products of sums of feature maps is not independent of the number of data points involved.\\
The contribution of this letter is to adopt the notion of kernel mean embedding to quantum mechanics, point out how quantum mechanics can lead to speedups, and make transparent what the challenges are in order to realize this in an experiment. \\
The letter is organized as follows. We start by introducing the kernel mean embedding from a classical perspective, point out the main problem it has in big data applications, and present its relevance in current machine learning research through some real-world applications. 
We then define the \textit{quantum mean embedding} as a modified version of the kernel mean embedding, which makes it suitable for investigation in the context of quantum computation, and show that this modification still allows for the usage in conventional applications.
We present how the quantum mean embedding can be used, in principle, to overcome the problems faced when working with the kernel mean embedding on a classical computer. 
Since this cannot be done on nowadays hardware, we discuss the necessary quantum routines in the CHALLENGES section. Finally, we sum up with a discussion of our results.

\vspace{-0.4em}
\section{Kernel mean embedding}
Let $\mathcal{X}$ be a locally compact and Hausdorff space.
A function $k: \mathcal{X} \times \mathcal{X} \to \mathbb{C}$  is called a positive definite kernel function, or kernel function for brevity, if for all $n\in  \mathbb{N}$, $x_1, ..., x_n \in \mathcal{X}$, and $c_1,...,c_n \in \mathbb{C}$, it holds that $\sum_{i,j=1}^n c_i^* c_j k(x_i,x_j) \geq 0$ \cite{Schoelkopf2001}. 
For every kernel function there exists a unique reproducing kernel Hilbert space (RKHS) $\mathcal{H}_k$ such that $k(\cdot, x)\in\HH_k$ for all $x\in\XX$ and the \emph{reproducing property} $f(x) = \langle f, k(\cdot, x)\rangle_{\HH_k}$ holds for all $f\in\HH_k$ and $x\in\XX$. We call the mapping $\phi: \mathcal{X} \to \mathcal{H}_k$ given by $\phi(x) := k(\cdot, x)$ the \emph{canonical} feature map of $k$, i.e., $k(x,y) = \braket{\phi(y), \phi(x)}$ \cite{Aronszajn1950}. \\ 
Let $\pp{P}$ be a probability measure over $\mathcal{X}$. The kernel mean embedding (KME) of $\mathbb{P}$ is defined as \cite{Berlinet2004, Smola2007}
\begin{align}\label{eq:kme}
    \mu_\mathbb{P} := \int_\mathcal{X} k( \cdot, x) \mathrm{d}\mathbb{P}(x) = \int_\mathcal{X} \phi(x) \mathrm{d}\mathbb{P}(x).
\end{align}
The embedding $\mu_{\pp{P}}$ exists and is a function in $\mathcal{H}_k$ if $\mathbb{E}_{X \sim \mathbb{P}}\left[k(X,X)\right] < \infty$ \cite{Smola2007}. For instance, $\mu_{\mathbb{P}} = \mathbb{E}_{X\sim\pp{P}}[X]$, i.e., the first moment of $\pp{P}$, if $k(x,y)=\langle x,y\rangle$. 
Higher-order moments of $\pp{P}$ can be incorporated via nonlinear kernel functions. \\ 
In practice we do not have access to the true probability distribution $\mathbb{P}$. Instead, we observe a finite i.i.d.~sample $X = \left\{x_1,...,x_n\right\}$ drawn from $\mathbb{P}$. 
Based on the sample $X$, an empirical estimate of the kernel mean $\mu_\mathbb{P}$ is  given by the KME of the empirical distribution $\hat{\mathbb{P}} = \frac{1}{n} \sum_{i=1}^n \delta_{x_i}$:
\begin{align}\label{eq:empKME}
    {\mu}_X := \frac{1}{n} \sum_{i=1}^n \phi(x_i),
\end{align}
which converges to the true embedding of $\mathbb{P}$ in the Hilbert space metric at a rate of $n^{-\frac{1}{2}}$ \cite{Muandet2017}. \\
The kernel function $k$ is said to be {\em characteristic} if the map $\mu: \mathbb{P}\mapsto \mu_\mathbb{P}$ is injective \cite{Fukumizu2008, Sriperumbudur2010}. In other words, working with a characteristic kernel enables us to represent (all properties of) a probability distribution by a function in the RKHS, which is why the notion of characteristic kernels plays an important role in kernel methods \cite{Simon-Gabriel}. The notion of characteristic kernels is closely related to the notion of universal kernels \cite{Steinwart2001}.
Here we call a kernel \textit{universal} if the corresponding RKHS is dense in the space of continuous functions over $\mathcal{X}$ that vanish at infinity, which corresponds to $c_0$-universality \cite{Simon-Gabriel}. Simon-Gabriel and Schölkopf \cite{Simon-Gabriel} show that for universal kernels, the embedding \eqref{eq:kme} is injective even when extended to finite signed measures.  Popular  kernels, which are universal, include the Gaussian kernel $k(x,y)=\exp(-\|x-y\|^2/2\sigma^2)$ and Laplacian kernel $k(x,y)=\exp(-\|x-y\|_1/\sigma)$, where $\sigma$ is a bandwidth parameter \cite{Fukumizu04:DRS,Sriperumbudur2010}. \\
The expressiveness of characteristic kernels comes at a price. 
Since there exist distributions with infinite moments, the corresponding RKHS must have infinite dimensions to ensure no information loss. 
Consequently, it is impossible for a classical computer to represent and manipulate $\mu_X$ directly. 
However, if we only care about inner products of mean embeddings, which is usually the case in most algorithms, we can resort to the ``kernel trick'' and replace inner products with kernel evaluations \cite{Schoelkopf2001}. 
That is, given i.i.d.~samples $X = \{x_1,\ldots,x_n\}$ from $\pp{P}$ and $Y = \{y_1,\ldots,y_n\}$ from $\pp{Q}$ \footnote{For simplicity we assume the sample sizes to be equal.}, we can evaluate 
\begin{align}\label{def:K(X,Y)}
    \braket{\mu_X, \mu_Y} &= \frac{1}{n^2} \sum_{i,j=1}^n \braket{\phi(x_i),\phi(y_j)} = \frac{1}{n^2} \sum_{i,j=1}^n k(x_i,y_j) \nonumber\\
    &=:K(X,Y).
\end{align}
The inevitable drawback of this trick is that algorithms based on $K(X,Y)$ have a runtime complexity that scales at least quadratically with the number of data points $n$. This is the limiting factor of the applications presented in the next section.

\vspace{-0.3em}

\subsection{Applications and Limitations}\label{sec:app_lim}

We highlight essential applications of the kernel mean embedding and the limitations of its use in classical computers.

\paragraph{Learning on Probability Distributions.}
Classical machine learning algorithms were originally developed for training data consisting of \emph{points} in some vector space. In several domains such as astronomy and high-energy physics, however, data are represented naturally as probability distributions, e.g., clusters of galaxies and groups of collision events. 
The kernel mean embedding \eqref{eq:kme} allows us to generalize these algorithms to the space of probability distributions \cite{SMM12,MuandetS2013,LopMuaSchTol2015,DR16}. For example, \cite{SMM12} proposed an algorithm called \emph{support measure machine} (SMM) which generalizes the SVM \cite{Cortes95:SN} to the space of probability distributions by means of the following kernel function
\begin{align}\label{eq:dist-kernel}
    K(\pp{P},\pp{Q}) = \langle \mu_{\pp{P}},\mu_{\pp{Q}}\rangle_{\HH_k} = \iint_{\XX} k(x,y) \dd\pp{P}(x)\dd\pp{Q}(y),
\end{align}
which is well defined over a space of probability distributions.
For certain classes of distributions and kernel functions, the kernel \eqref{eq:dist-kernel} can be evaluated analytically \cite[Table 1]{SMM12}.
This form of kernel function has been used extensively in many machine learning applications, see, e.g., \cite{Muandet2017} for a review.\\
Given i.i.d samples $X = \{x_1,\ldots,x_n\}$ from $\pp{P}$ and $Y = \{y_1,\ldots,y_n\}$ from $\pp{Q}$, $K(\pp{P},\pp{Q})$ can be approximated by
\begin{align}\label{eq:K(X,Y)}
    K(\pp{P},\pp{Q})
    \approx \frac{1}{n^2}\sum_{i,j=1}^nk(x_i,y_j) = K(X,Y).
\end{align}
The main drawback of \eqref{eq:K(X,Y)} is that, given samples $X_1,\ldots,X_N$ from $N$ input distributions, each of size $n$,  the runtime complexity of evaluating the kernels $K(X_i,X_j)$ for all $i,j=1,\ldots,N$ is $\mathcal{O}(N^2n^2)$. 
This is prohibitive for many real-world applications of learning problems on probability distributions.

\paragraph{Maximum Mean Discrepancy (MMD).}
The MMD is a discrepancy measure between any two distributions $\pp{P}$ and $\pp{Q}$ \cite{Borgwardt06:MMD,Gretton2012}. It is given by the distance of the corresponding mean embeddings of the distributions \cite[Lemma 4]{Gretton2012} and can be expressed solely in terms of inner products of mean embeddings (assuming a real kernel):
\begin{align}\label{eq:MMD}
\begin{aligned}
    \text{MMD}\left[\mathcal{H}_k, \mathbb{P}, \mathbb{Q}\right]^2&= \|\mu_\mathbb{P} - \mu_\mathbb{Q}\|^2 \\
    &= \braket{\mu_\mathbb{P}, \mu_\mathbb{P}} - 2 \braket{\mu_\mathbb{P}, \mu_\mathbb{Q}} + \braket{\mu_\mathbb{Q}, \mu_\mathbb{Q}}.
\end{aligned}
\end{align}
If the kernel is characteristic, the following implication holds: $\text{MMD}\left[\mathcal{H}_k, \mathbb{P}, \mathbb{Q}\right] = 0 \Leftrightarrow \mathbb{P} = \mathbb{Q}$ \cite[Theorem 5]{Gretton2012}.
Given i.i.d.~samples $X = \{x_1,...,x_n\}$ drawn from $\mathbb{P}$ and $Y = \{y_1,...,y_n\}$ drawn from $\mathbb{Q}$, it is possible to design a biased, but consistent,  estimator of the MMD by simply evaluating \eqref{eq:MMD} with the embeddings $\mu_X$ and $\mu_Y$ \cite[Eq.~(5)]{Gretton2012}. Here one uses the kernel trick to evaluate the inner products to get
\begin{align}\label{eq:biasMMD}
\begin{aligned}
    &\text{MMD}_b\left[\mathcal{H}_k, X, Y\right]^2 = \|\mu_X - \mu_Y\|^2 \\
    &= \frac{1}{n^2}\sum_{i,j=1}^n k(x_i,x_j) -\frac{2}{n^2}\sum_{i,j=1}^{n} k(x_i,y_j) +\frac{1}{n^2}\sum_{i,j=1}^n k(y_i,y_j) \\
    &= K(X,X) - 2 K(X,Y) + K(Y,Y),
    \end{aligned}
\end{align}
whose cost  is determined by that of evaluating  $K(X,X)$, $K(X,Y)$, and $K(Y,Y)$.

\paragraph{Deep Learning.} 

The applications of kernel mean embeddings in deep learning have gained a lot of attention in the past few years. 
Notably, the MMD has been used as an objective function for training deep generative models \cite{dziugaite2015training,li2015generative,li2017mmd}. 
For a deep generative model $G_{\theta}$ parametrized by a parameter vector $\theta$, the idea is to learn $\theta$ by minimizing the $\text{MMD}\left[\mathcal{H}_k, \mathbb{P}, \mathbb{Q}_{\theta}\right]^2$, where $\pp{P}$ is the data distribution and $\pp{Q}_{\theta}$ is the distribution induced by the generative model $G_{\theta}$.
Again, the downside of the MMD in this area is its computational cost as we usually have to deal with huge amount of data \cite{LeCun15:DL}.

\paragraph{Limitations.}
All of the above applications require the estimation of terms like $K(X,Y)$, which scale quadratically with the sample size $n$, and hence become prohibitive for large $n$.
To enable large-scale learning with kernel mean embeddings, a common approach is to approximate $\mu_X$ by a finite dimensional representation, e.g., using random Fourier features \cite{Rahimi2008} or the Nystr\"om method \cite{Williams2001}, after which it can be manipulated directly in a classical computer without resorting to the kernel trick. 
For a $d$ dimensional approximation, the cost drops to $\mathcal{O}(n+d)$, which is linear in $n$. The downside is that the embedding defined in terms of this representation can no longer be injective, which is an essential requirement in most applications of the KME. \\
Recent work \cite{Schuld2019, havlivcek2019} showed how one can in principle evaluate a $d$ dimensional approximation of the kernel function using only $\mathcal{O}(\log d)$ qubits. However, the quadratic scaling when using an infinite dimensional feature map has not been addressed so far in the quantum community.\\
In the next section we introduce the quantum mean embedding and show how this in principle allows us to explicitly work with the mean embedding even for an infinite dimensional feature map.

\vspace{-0.4em}

\section{Quantum mean embedding}
Let $\mathcal{H}$ be the Hilbert space of a quantum system and $\varphi: \mathcal{X} \to \mathcal{H}, x \mapsto \ket{\varphi(x)}$ a quantum feature map that assigns a quantum state $\ket{\varphi(x)}$, i.e., a normalized function in $\mathcal{H}$, to each point in the input domain $x \in \mathcal{X}$ 
\footnote{In order to emphasize that we deal with a quantum state, we shall abuse notation by denoting the image of a point $x$ under the mapping $\varphi$ as $\ket{\varphi(x)}$ instead of $\varphi(x)$. Mathematically $\ket{\varphi(x)}$ denotes the same function in $\mathcal{H}$ as $\varphi(x)$}. 
This defines a kernel $k(x,x') = \braket{\varphi(x)|\varphi(x')}$ \cite{Schuld2019, havlivcek2019} with the constraint $k(x,x) = 1$ for all $x \in \mathcal{X}$, due to the normalization of quantum states \cite{Nielsen2010}. \\
Let $\mathbb{P}$ be a probability distribution over the input domain. We define the \textit{quantum mean embedding} (QME) as
\begin{align}\label{eq:qe}
    \ket{\nu_\mathbb{P}} :=  \frac{1}{\mathcal{N}_\mathbb{P}}\int_\mathcal{X} \ket{\varphi(x)} \mathrm{d}\mathbb{P}(x),
\end{align}
where the normalization ${\mathcal{N}_\mathbb{P}}$ ensures the physicality of the state and is given by the norm of the corresponding kernel mean embedding \eqref{eq:kme}: \\
\begin{equation}
    \mathcal{N}_{\pp{P}}^2 := \|\mu_{\pp{P}}\|_{\HH_k}^2 = \iint_\mathcal{X} k(x,x') \dd\pp{P}(x)\dd\pp{P}(x').
\end{equation}
The QME exists for all probability distributions due to the constraint $k(x,x) =1$. 
A subtle difference between the KME and the QME are the spaces in which the embeddings live. While the KME is a function in the RKHS $\mathcal{H}_k$ and uniquely defined by the kernel $k$, the QME depends on the quantum systems Hilbert space $\mathcal{H}$ and the choice of the feature map $\varphi$.\\
Even though the embeddings  live in different spaces, for any two probability distributions $\pp{P}$ and $\pp{Q}$ we have
\begin{align}\label{eq:relationKME-QME}
    \braket{\mup ,\mu_\mathbb{Q}}_{\mathcal{H}_k} = \mathcal{N}_\mathbb{P}\cdot \mathcal{N}_\mathbb{Q} \braket{\nu_\mathbb{P}|\nu_\mathbb{Q}}_\mathcal{H}.
\end{align}
That is, their inner products have a fixed relation independent of $\mathcal{H}$.
Hence, the important difference is that the QME maps every probability distribution on the unit sphere in a Hilbert space, whereas the KME does not enforce this, see FIG.~\ref{fig:mapping}. In the following theorem we show that if the kernel is universal we do not lose information about a probability measure when using the QME.
\begin{figure}
    \centering
    \includegraphics[]{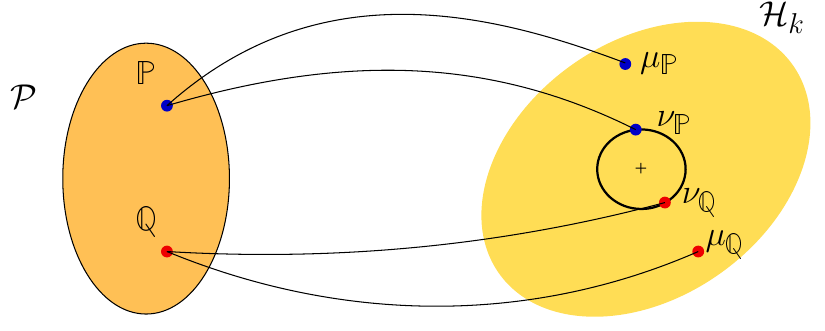}
    \caption{Schematic comparison of the classical KME and the QME: The KME maps probability distributions $\pp{P}$ onto functions in the RKHS $\mathcal{H}_k$. The QME additionally enforces that the mapping is onto the unit ball (denoted by the circle) in the RKHS. Theorem \ref{th:Injectivity} shows the injectivity of the QME for universal kernels. For visualization we choose $\mathcal{H} = \mathcal{H}_k$.}
    \label{fig:mapping}
\end{figure}
\begin{theorem}\textbf{Injectivity of the QME}\label{th:Injectivity}\\
Let $\varphi : \mathcal{X} \to \mathcal{H}, x \mapsto \ket{\varphi(x)}$ be a mapping such that $k(x,y) = \braket{\varphi(x)| \varphi(y)}$ is a universal kernel for the space of continuous functions over $\mathcal{X}$ that converge to zero at infinity $\mathscr{C}_0(\mathcal{X})$. Let $\mathcal{P}$ be the space of Borel probability measures  over the measurable space $(\mathcal{X},\mathcal{A})$, where $\mathcal{A}$ denotes the Borel sigma algebra. 
For a universal kernel $k$ the QME \eqref{eq:qe}, is injective over
 $\mathcal{P}$,
i.e., $\ket{\nu_\mathbb{P}} = \ket{\nu_\mathbb{Q}} \Leftrightarrow \mathbb{P} = \mathbb{Q} $ for any $\mathbb{P},\mathbb{Q} \in \mathcal{P}$.
\end{theorem}
\noindent The proof is included in the supplementary information.\\
For a finite sample $X$ we define an empirical QME as
\begin{align}\label{eq:empQEmb}
    \ket{{\nu}_X} := \frac{1}{\mathcal{N}_X} \frac{1}{n}\sum_{i=1}^n \ket{\varphi(x_i)},
\end{align}
where the normalization $\mathcal{N}_X$ is given by the norm of $\mu_X$ in  \eqref{eq:empKME}: 
\begin{equation}\label{eq:N_x}
     \mathcal{N}_{X}^2 := \|\mu_{X}\|_{\HH_k}^2 = \frac{1}{n^2}\sum_{i,j=1}^n k(x_i,x_j).
\end{equation}
As discussed before, for infinite dimensional feature maps, the KME cannot be described explicitly and only used via inner products. 
The advantage of the QME is that it is possible, in principle, to explicitly create $\ket{\nu_X}$ in the lab, even for infinite dimensional cases. 
Here it is important to note that an experimenter only needs to create a state that is proportional to $\sum_{i=1}^n \ket{\varphi(x_i)}$. 
The prefactor \eqref{eq:N_x} is enforced by the laws of physics and is not required for the state preparation. Given this explicit representation, it allows us to decouple the cost of the inner product evaluation from the sample size $n$, see FIG.~\ref{fig:approach}. 
\begin{conjecture}\label{conjecture}
    Suppose we are given a routine that prepares states of the form \eqref{eq:empQEmb} with cost $\mathcal{O}(n)$ for a feature map $\varphi$. In addition we are given a routine that can evaluate inner products of arbitrary states in $\mathcal{H}$ in constant time. Then for two samples $X = \{x_1,...,x_n\}$  and $Y = \{y_1,...,y_n\}$ one can evaluate 
    $K(X,Y)$, defined in \eqref{def:K(X,Y)},
    with cost $\mathcal{O}(n)$, whereas a classical computer scales with $\mathcal{O}(n^2)$.
\end{conjecture}
\begin{proof}
By assumption we can prepare $\ket{\nu_X}$ and $\ket{\nu_Y}$ with linear cost in $n$. Furthermore we can evaluate $\braket{\nu_X|\nu_Y}$ in constant time, given the individual states. Together the cost of evaluating the term $\braket{\nu_X|\nu_Y}$ scales at most with $\mathcal{O}(n)$. 
The normalizations $\mathcal{N}_X$ and $\mathcal{N}_Y$ can also be estimated with cost $\mathcal{O}(n)$, see CHALLENGES section. Using the relation \eqref{eq:relationKME-QME} between the KME and the QME, we can calculate
\begin{align}\label{eq:K(X,Y)_quantum}
    K(X,Y) = \braket{\mu_X, \mu_Y}_{\mathcal{H}_k} = \mathcal{N}_X\mathcal{N}_Y \braket{\nu_X|\nu_Y}_\mathcal{H}.
\end{align}
 \end{proof}
Given an efficient evaluation of $K(X,Y)$, it is possible to speed up the methods presented earlier, which rely on inner products of the KMEs.
We discuss the assumptions of Conjecture 1 in the CHALLENGES section.\\
Apart from using the QME to speed up the evaluation of inner products of the KMEs, it follows from the proof of Theorem \ref{th:Injectivity} that the QME is also important on its own, as it can uniquely represent probability distributions. However, it is unclear to what extend the applications of the KME could be rephrased solely in terms of inner products of the QME instead of taking the detour over the $K(X,Y)$, where we additionally need to determine the normalizations. This has not been investigated in the machine learning literature so far.

\vspace{-0.4em}

\section{Challenges}\label{sec:Challenges}
 \begin{figure}
    \centering
    \includegraphics[]{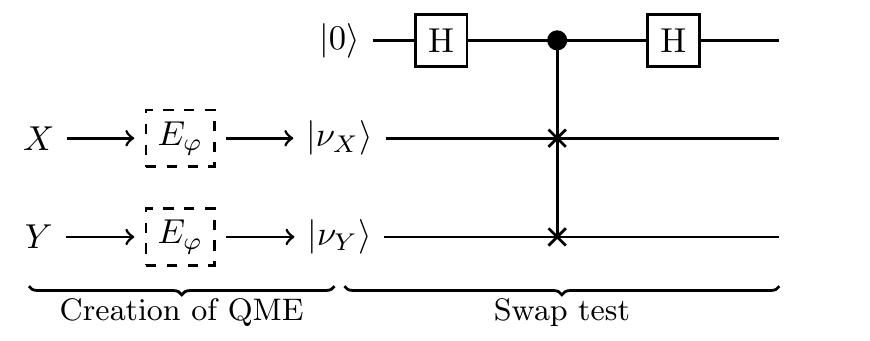}
    \caption{The quantum approach separates the creation of the QME from the inner product estimation. It requires two subroutines. First, on the left, an experimental setup $E_\varphi$ that creates the QME efficiently. Second, on the right, a circuit to estimate inner products of arbitrary states in $\mathcal{H}$ whose runtime is independent of the states. Here we chose the swap test, which uses an ancillary qubit. This approach detaches the estimation of the inner product from the sample size. }
    \label{fig:approach}
\end{figure}

\vspace{-0.4em}

\subsection{Efficient preparation of $\ket{\nu_X}$}
In order to harvest a potential quantum speedup it is necessary to create the QME efficiently, i.e., with resources and time linear in the sample size. We phrase this as the first challenge:\\
\textit{Given a quantum feature map $\varphi$, find an experimental strategy, denoted $E_\varphi$, such that for an  arbitrary input sample $X = \{x_1,...,x_n\}$, with $n\in \mathbb{N}$, it creates $\ket{\nu_X}$, using resources that scale at most linear in $n$.}\\
In case of coherent states as feature map (see supplementary information), superpositions similar to $\ket{\nu_X}$ have already been experimentally realized for specific cases and are known as ``cat-states''  \cite{Vlastakis2013, Deleglise2008, Ourjoumtsev2007}. However, it is an open question how these approaches scale, even theoretically, for superposing a large number of states, see \cite{Andersen2015} for an overview on similar experimental approaches. 
In general, the rigorous study of resources required to construct superpositions of quantum states and the connections to entanglement are subject of current research \cite{Theurer2017, Streltsov2017}. Particularly for the case of superpositions of nonorthogonal states, as it is the case for our proposed embedding, the theory becomes more involved, see III.K.4.~of \cite{Streltsov2017}. \\
Note that we explicitly allow for an experimental setup $E_\varphi$ that is specific to the given quantum feature map $\varphi$, i.e., a specific kernel function. This is necessary because a universal machine that builds a superposition of completely arbitrary and unknown quantum states cannot exist \cite{Alvarez-Rodriguez2015, Oszmaniec2016}. Furthermore, we  emphasize that this work does not require a quantum random access memory (qRAM) \cite{Giovannetti2008}.
Firstly, because a qRAM would entangle each single state with a state of an additional data register, whereas in this work a pure superposition is required without any entanglement. Secondly, we aim for a polynomial speedup, hence there is no necessity for a logarithmic scaling of the preparation. %

\vspace{-0.4em}

\subsection{Estimation of inner products}
At the core of the advantage in using the QME is the estimation of the inner product of two arbitrary quantum states in $\mathcal{H}$. Formally, this can be done by using the \textit{swap test} routine of \cite{Buhrman2001}, see right side of FIG.~\ref{fig:approach}. The swap test works independently of the input states, which for our purpose we denote by $\ket{\nu_X}, \ket{\nu_Y} \in \mathcal{H}$. These inputs are each in one register and a single ancilla qubit in the state $\ket{0}$ in an additional register. The test itself consists of a Hadamard transformation $H$ on the qubit, followed by a controlled swap of the two states conditioned on the state of the qubit, and another Hadamard transformation on the qubit. This circuit maps the initial state $\ket{0}\ket{\nu_X}\ket{\nu_Y}$ onto 
\begin{align}
    \frac{\ket{0}\left( \ket{\nu_Y}\ket{\nu_X} + \ket{\nu_X}\ket{\nu_Y}\right) + \ket{1}\left( \ket{\nu_Y}\ket{\nu_X} - \ket{\nu_X}\ket{\nu_Y}\right)}{2},
\end{align}
see \cite[Eq.~(4)]{Buhrman2001}. At the end, the qubit is measured in the computational basis. This results in outcome $1$ with probability $p_1 = (1-|\braket{\nu_X|\nu_Y}|^2)/{2}$ and outcome $0$ with probability $p_0 = (1+|\braket{\nu_X|\nu_Y}|^2)/{2}$.
Repetitive application of this routine allows for an estimation of $p_0$ and $p_1$ from which one can infer  $|\braket{\nu_X|\nu_Y}|^2 = 2 p_0 -1$. When using a Gaussian kernel, we know a priori that $\braket{\nu_X|\nu_Y} > 0$, thus $\braket{\nu_X|\nu_Y} = \sqrt{2 p_0 -1}$. 
If we cannot guarantee the positivity of $\braket{\nu_X|\nu_Y}$, we need a phase sensitive estimation of inner products, as discussed in the supplementary information of \cite{Schuld2019}. Crucially, the swap test works independently of the size of the samples $X$ and $Y$.\\
For finite dimensional systems, Cincio et al.~\cite{Cincio2018} recently proposed an implementation that scales linearly with the number of qubits and hence logarithmically with the dimension of the Hilbert space. 
But this approach does not translate to systems of infinite dimension. The infinite dimensional case has been studied in \cite{Filip2002, Pregnell2006, Jeong2014}. However, they do not give an explicit solution and we are not aware of any experimental realization of a universal swap test for the infinite dimensional case. This marks the second challenge arising from this letter.

\vspace{-0.4em}

\subsection{Estimation of the normalization $\mathcal{N}_X$}
At the stage of preparing superpositions in the form of \eqref{eq:empQEmb} on a quantum device, it is not necessary to know the value of the normalization $\mathcal{N}_X$. However, if the goal is to estimate $K(X,Y)$ with the help of a quantum device, then knowledge of the normalizations is needed, see \eqref{eq:K(X,Y)_quantum}. The naive approach using its definition \eqref{eq:N_x}, takes $\mathcal{O}(n^2)$ operations and would prohibit  the polynomial advantage. 
\\
To evade this, we can evaluate $\mathcal{N}_X$ by estimating the inner product with a reference state $\ket{\psi_\text{ref}} = \ket{\varphi(x_\text{ref})}$ for some reference value $x_\text{ref}\in \mathcal{X}$. To this end, we analytically calculate
\begin{align}
    c := \frac{1}{n} \sum_{i=1}^n  \braket{\psi_\text{ref}|\varphi(x_i)} = \frac{1}{n} \sum_{i=1}^n k(x_\text{ref} , x_i) , 
\end{align}
using $\mathcal{O}(n)$ operations. 
Now given the preparation of $\ket{\nu_X}$ and of $\ket{\psi_\text{ref}}$ we can experimentally evaluate the inner product $\braket{\psi_\text{ref}|\nu_X}$  and from this  obtain the normalization $
   \mathcal{N}_X =  c{\braket{\psi_\text{ref}|\nu_X}}^{-1}$.
 Obviously, in order to make this well defined, we need to choose the reference function such that $\braket{\psi_\text{ref}|\nu_X} \neq 0$.
This strategy relies on the challenges phrased in the previous two paragraphs but apart from this does not pose an extra difficulty by itself. \\
We emphasize again that due to Theorem \ref{th:Injectivity} it should be possible to come up with algorithms that directly work with the QME and hence make the estimation of the normalization superfluous.

\vspace{-0.4em}

\section{Conclusion}
In this work, we adapted the concept of kernel mean embeddings to quantum mechanics, by defining what we call quantum mean embedding. While the kernel mean embedding maps a probability distribution to a function in a reproducing kernel Hilbert space, the quantum mean embedding can only map onto the unit sphere of a Hilbert space, a necessity that arises due to the normalization of quantum states. Despite this additional constraint, we showed that the quantum mean embedding is still injective if the induced kernel is universal. Since the quantum mean embedding can, in principle, be created in the lab, it allows for a polynomial speedup when computing inner products between mean embeddings of empirical distributions. We highlighted the relevance of this task by describing use cases in recent machine learning applications. We made explicit which requirements need to be fulfilled by the quantum hardware in order to harvest the polynomial advantage.\\
This work opens multiple paths for further research. On the quantum side, the experimental creation of superpositions of a large number of states and the estimation of inner products thereof. Furthermore, the quantum mean embedding is a new way of encoding probability distributions in quantum states, which allows us to use the results known from the kernel theory. For machine learning research, it is an open question what the possible applications of the  embedding of probability distributions onto the unit sphere in the reproducing kernel Hilbert space could be.\\

J.K.~would like to thank C.J.~Simon-Gabriel for his advice on universal and characteristic kernels.

\bibliography{library}


\vspace{-0.5em}

\section{Supplementary information}

\begin{proof}[Proof of Theorem \ref{th:Injectivity}]
We make the proof in terms of the canonical feature map $\phi$, which maps into the RKHS. The validity for any mapping $\varphi: \mathcal{X} \to \mathcal{H}$ that leads to the same kernel function is then trivial.\\
Let $\mathcal{M}(\mathcal{X},\mathcal{A})$ denote the set of finite non-negative measures on the measurable space $(\mathcal{X},\mathcal{A})$, i.e.,  $\xi (\mathcal{X}) < \infty$ for all $\xi \in \mathcal{M}(\mathcal{X},\mathcal{A})$.
We can extend the definition of the kernel mean embedding \eqref{eq:kme} to $\mathcal{M}(\mathcal{X},\mathcal{A})$ by defining
\begin{align}\label{eq:MeasureEmb}
    \mu_\mathbb{\xi} = \int_\mathcal{X} k( \cdot, x) \mathrm{d}\mathbb{\xi}(x) = \int_\mathcal{X} \phi(x) \mathrm{d}\mathbb{\xi}(x),
\end{align}
for any $\xi \in \mathcal{M}(\mathcal{X},\mathcal{A})$ that fulfills $\int_\mathcal{X} k(x,x) \mathrm{d}\mathbb{\xi}(x) < \infty$.
Let $\xi_1$ and $\xi_2$ be arbitrary measures in $\mathcal{M}(\mathcal{X},\mathcal{A})$. 
By assumption, $k$ is universal over $\mathscr{C}_0(\mathcal{X})$ and thus characteristic over $\mathcal{M}(\mathcal{X},\mathcal{A})$, i.e., $\mu_{\xi_1} = \mu_{\xi_2} \Leftrightarrow \xi_1 = \xi_2$; see Theorem 6 in \cite{Simon-Gabriel}.\\
Define $\nu_\mathbb{P}$ as the mean embedding onto the unit sphere of the RKHS 
\begin{align}\label{eq:nu_RKHS}
    \nu_\mathbb{P} := \frac{1}{\mathcal{N}_\mathbb{P}} \mu_\mathbb{P},
\end{align}
with $\mathcal{N}_\mathbb{P} \in \mathbb{R}^+$ such that $\|\nu_\mathbb{P}\|_{\mathcal{H}_k} = 1$.
Let $\mathbb{P}$ and $\mathbb{Q}$ be probability measures for which the embedding onto the unit sphere \eqref{eq:nu_RKHS} coincide, i.e., $\nu_\mathbb{P} = \nu_\mathbb{Q}$. We can relate this to the kernel mean embeddings as
\begin{align}
    \mu_\mathbb{P} = \mathcal{N}_\mathbb{P} \, \nu_\mathbb{Q} = \frac{\mathcal{N}_\mathbb{P}}{\mathcal{N}_\mathbb{Q}} \mu_\mathbb{Q} = \mu_\xi,
\end{align}
where we defined the finite non-negative measure $\xi = \frac{\mathcal{N}_\mathbb{P}}{\mathcal{N}_\mathbb{Q}} \mathbb{Q}$, using the linearity of \eqref{eq:MeasureEmb}. 
With the injectivity of the embedding \eqref{eq:MeasureEmb} this implies $\mathbb{P} =\xi =\frac{\mathcal{N}_\mathbb{P}}{\mathcal{N}_\mathbb{Q}} \mathbb{Q}$. 
By assumption, $\mathbb{P}$ and $\mathbb{Q}$ are probability measures and fulfill $\mathbb{P}(\mathcal{X}) = \mathbb{Q}(\mathcal{X}) = 1$. This implies $\frac{\mathcal{N}_\mathbb{P}}{\mathcal{N}_\mathbb{Q}} = 1$ and thus $\mathbb{P} = \mathbb{Q}$, which proves the injectivity of $\nu$ for the set of probability distributions.
\end{proof}

\subsection{Coherent states and Gaussian kernel}
In this section, we consider an explicit example, previously reported in \cite{Chatterjee2017}. Let $\mathcal{H}$ be an infinite dimensional (complex) Hilbert space, with orthonormal basis $\{\ket{n}\}_{n \in \mathbb{N}_0}$. This could for example be the space corresponding to a single mode of the electro-magnetic field \cite{Strekalov2019}. For simplicity we consider $\mathcal{X} = \mathbb{R}$ and define the feature map $\varphi: \mathbb{R} \to \mathcal{H}$ as
\begin{align}\label{eq:coherent}
    \ket{\varphi(x)} = e^{-\frac{1}{2} x^2} \sum_{n=0}^\infty \frac{x^n}{\sqrt{n!}}\ket{n}.
\end{align}
In quantum optics, the states $\ket{n}$ are called Fock states. 
States of the form \eqref{eq:coherent} are called coherent states and are well studied \cite{agarwal2012}. 
In the context of this paper, however, the nature of the basis and hence the exact form of the Hilbert space are unimportant. 
The important part is the orthonormality of the basis states, which implies 
\begin{align}
    \braket{\varphi(x)|\varphi(x')} = e^{-\frac{1}{2}(x-x')^2} =: k(x,x'),
\end{align}
for arbitrary $x, x'\in \mathbb{R}$ and defines the popular Gaussian kernel \cite{Schoelkopf2001}. 
By composing the mapping \eqref{eq:coherent} with the mapping $x \mapsto \frac{x}{\sigma}$, for some $\sigma > 0$, it is also possible to include a bandwidth parameter $\sigma$. 
The Gaussian kernel fulfills the requirements of Theorem \ref{th:Injectivity} (see, \cite[theorem 17]{Simon-Gabriel}). 
Therefore, it is possible to construct an injective embedding of probability distributions over the real numbers in a superposition of coherent states.\\
Coherent states are commonly considered the \textit{most classical} states in quantum optics, and are easy to simulate on a classical device. 
Working with a quantum device becomes interesting when the states become \textit{nonclassical} \cite{Strekalov2019}. 
When using the coherent feature map \eqref{eq:coherent}, the embedding of a sample \eqref{eq:empQEmb} corresponds to the so-called \textit{cat-states} \cite{Vlastakis2013, Deleglise2008, Ourjoumtsev2007}. Cat-states are considered nonclassical, as their Wigner function attains negative values. 
From a quantum perspective, this already hints to the difficulties encountered when working with such states on a classical devices.
\end{document}